\documentclass[preprint,12pt]{elsarticle}

\usepackage{latexsym}
\usepackage{amssymb}
\usepackage[normalem]{ulem}

\usepackage{amssymb,amsthm}
\usepackage{amsmath}
\usepackage{amsfonts}
\usepackage{ulem}
\usepackage{tikz}
\usetikzlibrary{decorations.pathreplacing,snakes,mindmap}
\usetikzlibrary{arrows, calc, decorations.pathmorphing}

\newtheorem{theorem}{Theorem}
\newtheorem{lemma}[theorem]{Lemma}
\newtheorem{property}[theorem]{Property}
\newtheorem{corollary}{Corollary}[theorem]

\begin{document}

\begin{frontmatter}

\title{Structural and Spectral Properties of Strictly Interval Graphs}

\author{Claudia M. Justel$^{1}$} \author{Lilian Markenzon$^{2}$}

\address{$^{1}$ PGSC - Departamento de Engenharia de Computa\c c\~ao, Instituto Militar de Engenharia, Rio de Janeiro, Brazil.
}
\address{$^{2}$ PPGI - Universidade Federal do Rio de Janeiro, Rio de Janeiro, Brazil.
}

\begin{abstract}
{In this paper we deal with a subclass of chordal graphs, which are simultaneously strictly chordal and interval, the \textit{strictly interval graphs}.
We present a new characterization of the class that
leads to a simple linear recognition algorithm.
Next we  introduce a new subclass of strictly interval graphs, the $SI$-core graphs, 
that are  non-split and non-cograph graphs and show that several elements of the new class are Laplacian integral.}
\end{abstract}


\end{frontmatter}

%%%%%%%%%%%%%%%%%%%%%%%%%%%%%%%%%%%%%%%%%%%%%%%%%%%%%%%%%%%%%%%%%
%%%%%%   introduction
%%%%%%%%%%%%%%%%%%%%%%%%%%%%%%%%%%%%%%%%%%%%%%%%%%%%%%%%%%%%%%%%%

\section{Introduction}\label{intro}

Strictly chordal graphs (also called block duplicate graphs), a subclass of chordal graphs, were first introduced by Golumbic and Peled \cite{GP02}; in the paper the authors proved that these graphs are (gem,dart)-free.
%in that paper a forbidden subgraph characterization was given by the authors. 
The same class  was defined in terms of hypergraphs by Kennedy \cite{K05}. 
In 2010, Brandst\"adt and Wagner \cite{BW10} characterized strictly chordal graphs by their critical clique graph structure.    
That characterization led to several structural properties of the class.
In 2015, a characterization based on minimal
vertex separators was presented by Markenzon and Waga \cite{MW15}.
Here we deal with a subclass of strictly chordal graphs 
which are also interval graphs,
the \textit{strictly interval graphs}.
They were introduced by Markenzon and Waga \cite{MW16},
who proved that they are (2-net, bipartite claw)-free (see Figure \ref{fig:forbidden}).
We present a new characterization of  strictly interval graphs using their critical clique graphs 
leading to a simpler recognition algorithm.
We  also introduce a subclass of strictly interval graphs, the  $SI$-core graphs.

We focus, particularly, on the relation between classic invariants of graphs and 
their integer Laplacian eigenvalues.
The quest for classes of graphs with elements
which are Laplacian integral graphs is widely studied in the literature, being 
the cographs the most important example of these classes.
Based on the structure of the $SI$-core graphs we show that several 
elements of the class are  Laplacian integral.

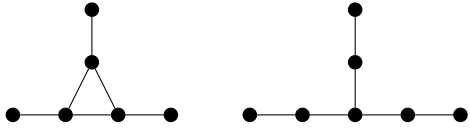
\begin{figure}[h]
\begin{center}
\begin{tikzpicture} [scale=.35,auto=left]
 \tikzstyle{every node}=[circle, draw, fill=black,
                        inner sep=1.8pt, minimum width=4pt]

  \node (a) at (2,2) {};
  \node (b) at (4,2) {};
  \node (c) at (6,2)  {};
  \node (d) at (8,2)  {};
  \node (e) at (5,4) {};
  \node (f) at (5,6) {};

\foreach \from/\to in {a/b,b/c,c/d,b/e,c/e,f/e}
    \draw (\from) -- (\to);

 \node (g) at (11,2) {};
  \node (h) at (13,2) {};
  \node (i) at (15,2)  {};
  \node (j) at (17,2)  {};
  \node (k) at (19,2) {};
\node (m) at (15,4) {};
\node (n) at (15,6) {};

\foreach \from/\to in {g/h,h/i,i/j,j/k,i/m,m/n}
    \draw (\from) -- (\to);
     \end{tikzpicture}
\caption{2-net and bipartite claw graphs}
  \label{fig:forbidden}
\end{center}
\end{figure}

%%%%%%%%%%%%%%%%%%%%%%%%%%%%%%%%%%%%%%%%%%%%%%%%%%%%%%%%%%%%%%%%%
%%%%%%   basic notions
%%%%%%%%%%%%%%%%%%%%%%%%%%%%%%%%%%%%%%%%%%%%%%%%%%%%%%%%%%%%%%%%%

\section{Basic Graph Concepts}

All graphs mentioned in this paper are connected graphs.

Basic concepts about chordal graphs are assumed to be known and
can be found in Blair and Peyton \cite{BP93} and Golumbic \cite{Go04}.
Let $G = (V, E)$ be a graph, with $|E|=m$,  $|V| = n > 0$.
The \textit{set of neighbors\/} of a vertex $v \in V$ is denoted by
$N(v) = \{ w \in V \mid \{v,w\} \in E\}$ and the \textit{closed neighborhood} of a vertex $v$ is defined by $N[v] = N(v)\cup \{v\} $.
Two vertices $u$ and $v$   are  \textit{true twins} in $G$  if $N[u] = N [v]$ and 
they are \textit{false twins}  in $G$  if $N (u) = N(v)$.
A vertex $v$ is said to be \textit{simplicial\/} in $G$ if $N(v)$  is a
clique in $G$.
%the neighborhood of a simplicial vertex is called the {\em simplicial neighborhood}.
Given a chordal graph $G$, the \textit{derived graph} $D(G)$ 
is the induced subgraph obtained by deleting all 
the simplicial vertices of $G$.

A subset $S \subset V$  is a \textit{vertex separator}  for
non-adjacent vertices $u$  and $v$  (a $uv$-separator) if the
removal of $S$ from the graph separates $u$ and $v$  into distinct
connected components. If no proper subset of $S$  is a
$uv$-separator then $S$ is a \textit{minimal $uv$-separator}. When the
pair of vertices remains unspecified, we refer to $S$  as a \textit{minimal vertex separator}, denoted $mvs$. 

A \textit{clique-tree} of $G$ is defined as a tree $T$  whose vertices 
are the maximal cliques of $G$ such that for every  two maximal cliques $Q$ and $Q^\prime$
each clique in the path from $Q$ to $Q^\prime$ in $T$ contains $Q\cap Q^\prime$. 
The set of maximal cliques of $G$ is denoted by $\mathbb Q$.
The \textit{leafage} of a chordal graph $G$, $\ell(G)$,  is the minimum number of leaves among
the clique-trees of $G$.
Blair and Peyton \cite{BP93}   proved that, 
for a clique-tree $T=(V_T, E_T)$,
a set $S\subset V$ is a minimal vertex separator of $G$ if
and only if $S= Q' \cap Q''$ for some edge $\{Q', Q''\}\in E_T$; $S$ can appear several times in $T$.
Moreover, the multiset  ${\mathbb M}$ of
the minimal vertex separators of $G$ is the same for every
clique-tree of $G$.
The multiplicity of the minimal vertex separator $S$, denoted by $\mu(S)$, is the number of times that
$S$ appears in $\mathbb M$.
The set of minimal vertex separators of $G$ is denoted by $\mathbb S$. 
The determination of
the minimal vertex separators and their multiplicities can be performed in linear time (Markenzon
and Pereira  \cite{MP10}).

It is worth mentioning special kinds of cliques.
A \textit{simplicial clique} is a maximal clique containing at least one simplicial vertex.
A simplicial vertex appears in a single maximal clique of the graph.
A simplicial clique $Q$ is called a \textit{boundary clique} if there exists a maximal clique $Q^\prime$
such that  $Q \cap Q^\prime=X $ is the set of non-simplicial vertices of $Q$. 
Shibata \cite{S88} showed that if a maximal clique $Q$ is a leaf in some clique-tree 
then $Q$ is a boundary clique.
The converse of this fact was proved by
Hara and Takemura \cite{HT06}.
These results can be summarized in the next theorem.

\begin{theorem}\label{theo:leaf-boundary} {\rm\cite{HT06}}
A maximal clique is a leaf in some clique-tree if and only if it is a boundary clique.
\end{theorem}

A clique such that all vertices are true twins and no vertex of the graph can be added to the clique without violating this property is called a {\it critical clique}.  
The set $\mathbb{C}$ of critical cliques of $G$  induces a partition of the vertex set $V$.

The \textit{critical clique graph} of a chordal graph $G$, $CC(G)$,  is the
connected  graph whose vertices are the critical cliques of $G$ and whose
edges connect vertices corresponding to  cliques $Q_i$ and $Q_j$ such that
some vertex of $Q_i$ is adjacent to a vertex of $Q_j$.
Note that $CC(G)$ has no true twins \cite{BW10}.  The maximal cliques of $CC(G)$ have at most one
simplicial vertex.

%%%%%%%%%%%%%%%%%%%%%%%%%%%%%
% SUBCLASSES
%%%%%%%%%%%%%%%%%%%%%%%%%%%%%%

\section{Subclasses of chordal graphs}

In this section we present the definitions and some characterizations of the subclasses of chordal graphs 
that will be used in this paper.

An {\it interval graph} is  a graph where each vertex can be associated with an interval on the real line, 
such that two vertices are adjacent if and only if their corresponding intervals have a nonempty intersection. 
There are several characterizations of this class. 
The following characterization theorem guarantees that for an interval graph,
 its maximal cliques can be linearly ordered.

\begin{theorem}\label{theo:charact-IG2}  {\rm\cite{GH64}}
Let  $G=(V,E)$ be a non-complete chordal  graph. Then, 
$G$ is an interval graph if and only if $G$ has a  clique-tree that
is a  path. 
\end{theorem}

\textit{Strictly chordal graphs}, defined in \cite{K05} based on hypergraph properties,  
were firstly introduced as \textit{block duplicate graphs}  by  Golumbic and Peled \cite{GP02}.
A \textit{strictly chordal  graph} is a graph obtained by adding zero or more true twins 
to each vertex of a block graph $G$ 
(a graph is a \textit{block graph} \cite{H63}  if it is connected and every block is a clique).
 The  class was proved to be chordal,  gem-free and dart-free (\cite{GP02}, \cite{K05}), see Figure \ref{fig:gem-dart}.

\bigskip

\begin{figure}[h]
\begin{center}
\begin{tikzpicture}
 [scale=.43,auto=left]
 \tikzstyle{every node}=[circle, draw, fill=black,
                        inner sep=1.8pt, minimum width=4pt]

  \node (c) at (1,10) {};
  \node (d) at (3,10) {};
  \node (b) at (-0.5,8.5)  {};
  \node (e) at (4.5,8.5)  {};
  \node (a) at (2,7) {};

  \foreach \from/\to in {c/d,c/b,c/a,d/a,d/e,b/a,e/a}
    \draw (\from) -- (\to);

  \node (n) at (14,10) {};
  \node (k) at (11.5,8.5) {};
 \node (l) at (14,8.5)  {};
  \node (m) at (16.5,8.5)  {};
  \node (j) at (14,7)  {};
  \foreach \from/\to in {n/k,n/l, k/l, l/m,k/j, l/j} 
    \draw (\from) -- (\to);

    \end{tikzpicture}
\end{center}
\caption{Gem and dart  graphs.} 
\label{fig:gem-dart}
\end{figure}
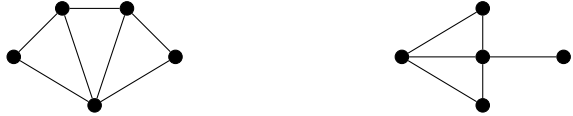

\begin{theorem} {\rm\cite{BW10}}\label{theo:BW}
Let $G$ be a connected graph. 
$G$ is strictly chordal if and only if $CC(G)$ is a block graph.
\end{theorem}

 Strictly chordal graphs  can also be characterized in terms of the structure of their minimal vertex separators. 

\begin{theorem}\label{theo:caract-BD}{\rm \cite{MW15}}
Let  $G=(V,E)$ be a chordal graph.  
The following statements are equivalent: 
\begin{enumerate}
  \item $G$ is a strictly chordal graph.
  \item For any distinct $S, S^{\prime} \in {\mathbb S}$, $S \cap S^{\prime}= \emptyset$.
  \item $G$ is a $\{\textrm{gem, dart}\}$-free graph. 
\end{enumerate}
\end{theorem}

  \begin{corollary}  \label{prop:sc}  Let $G=(V,E)$ be  a strictly chordal graph and $\mathbb S$ its set of minimal vertex separators. The following holds:
 \begin{enumerate}
 
 \item[(a)]  $S \in {\mathbb S}$  is a minimal separator of $G$.

  \item[(b)] $S$ belongs to exactly $\mu(S)+1$ maximal cliques of $G$, for every  $S \in {\mathbb S}$.
    
 \item[(c)] boundary cliques of $G$ contain only one minimal vertex separator. 

  \end{enumerate}
 \end{corollary}

Based on this characterization, the recognition algorithm for strictly chordal graphs is extremely simple and relies only on the determination of $\mathbb S$.

The critical clique graph of strictly chordal graphs has interesting properties.

By Theorem \ref{theo:caract-BD}, the proof of the next lemma becomes immediate. 

\begin{lemma}\label{lem:vertices-cliques}
Let $G$ be a strictly chordal graph.
A critical clique of $G$ consists of:
\begin{enumerate}
\item all the simplicial vertices of a maximal clique of $G$ or
\item all the vertices of a minimal vertex separator of $G$.
\end{enumerate}
\end{lemma}

\begin{proof}
Let $G$ be a strictly chordal graph.
A maximal clique $Q$ consists of simplicial vertices and vertices belonging to
minimal vertex separators.
A simplicial vertex belongs to a sole maximal clique; all simplicial vertices of $Q$ are true twins.  
So, they induce a critical clique. By Theorem \ref{theo:caract-BD} (2) for any distinct $S, S' \in {\mathbb S},  S \cap S'= \emptyset$.
Then all vertices of $S$ have the same neighborhood and they induce a critical clique of $G$. 
\end{proof}

%%%%%%%%%%%%%%%%%%%%%%%%%%%%%%%%%%%%%%%%%%%%%%%%%%%%%%%%%%%%%%%%%
%%%%%%   SIG
%%%%%%%%%%%%%%%%%%%%%%%%%%%%%%%%%%%%%%%%%%%%%%%%%%%%%%%%%%%%%%%%%

\section{Strictly Interval Graphs: new characterization and recognition algorithm}

In this section we come back to a subclass of interval graphs, 
the \textit{strictly interval graphs} ({\em SIG}), graphs which are at the same time
strictly chordal and interval graphs.
A first characterization of the  class was presented in \cite{MW16}.
We present structural results about the class,
leading to a new characterization and a simple recognition algorithm.

\begin{theorem}\label{theo:charact-SIG1}{\rm \cite{MW16}}
Let  $G=(V,E)$ be a non-complete strictly chordal graph.
The following statements are equivalent: 
\begin{enumerate}
\item $G$ is a strictly interval graph.
\item $G$ is a $\{\textrm{2-net, bipartite claw}\}$-free graph (see Fig. \ref{fig:forbidden}).  
 \item  leafage$(G) =2$.
\end{enumerate}
\end{theorem} 

\begin{corollary}{\rm \label{cor:mc}}
Let $G$ be a non-complete strictly interval graph. 
Then every maximal clique of $G$ contains at most two minimal vertex separators.
\end{corollary}

\begin{proof}
By  Theorem \ref{theo:caract-BD}, %Theorem 4.1, 
if $G$ is a strictly interval graph then for any distinct 
$S, S^{\prime} \in {\mathbb S}$, $S \cap S^{\prime}= \emptyset$.
If a maximal clique   contains  three minimal vertices separators,
graph $G$ contains a 2-net graph, a forbidden subgraph for strictly interval graphs.
\end{proof}

As it was seen in Theorem \ref{theo:BW}, strictly chordal graphs are characterized 
by their critical clique graphs.
The next theorem presents  a new characterization of strictly interval graphs,
also based in their critical clique graphs,  that
supports a more efficient recognition algorithm.

\begin{theorem}\label{theo:charact-SIG2} 
Let $G$ be a non-complete strictly chordal graph. 
  $G$ is a   strictly interval graph if and only if 
$D(CC(G))$ is a path.
\end{theorem}

\begin{proof}
\bigskip
$\Longrightarrow$  
Let $G$ be a strictly interval graph.
By Theorem \ref{theo:BW}, $CC(G)$ is a block graph.
As $G$ is connected, $CC(G)$ is also connected.
By Lemma \ref{lem:vertices-cliques} all vertices of a minimal vertex separator induces a critical clique.
So, each $mvs$ of $G$ is represented by one vertex in $CC(G)$.
Also by Lemma \ref{lem:vertices-cliques}, the simplicial vertices belonging to the same maximal clique of $G$
induce a critical clique in $CC(G)$; this critical clique is a simplicial vertex of $CC(G)$.
These vertices are removed in $D(CC(G))$.
Only critical cliques corresponding to minimal vertices separators remain in $D(CC(G)$.
As only simplicial vertices are removed, $D(CC(G))$ is connected.

Suppose that $S_1$, $S_2$ and $S_3$ are $mvs$ of $G$, corresponding to vertices $s_1,s_2,s_3$ disposed as :

\begin{enumerate}
\item edges $(s_1,s_2)$ and $(s_2,s_3)$, with no edge $(s_1,s_3)$, forming a path in $D(CC(G))$;
\item edges $(s_1,s_2)$, $(s_2,s_3)$ and $(s_1,s_3)$, forming a triangle in $D(CC(G))$.
Observe that each of these vertices separates cliques in $CC(G)$, 
that correspond to critcal cliques of $G$;
see Figure \ref{fig:theo-charact}(a) and (c).
In this case $CC(G)$ (and $G$) has a 2-net (see Figure \ref{fig:theo-charact}(a) and (c)).
\end{enumerate}

Let us now consider a vertex $s'$ that does not belong to this path but, together with $P$,
produces a connected induced subgraph  of $D(CC(G))$.
Vertex $s'$ can be linked to $P$ in three ways:

\begin{enumerate}
 \item[3.] edge $(s',s_1)$  (or $(s',s_3))$: in this case it forms a path in $D(CC(G))$;
 \item[4.] edge $(s',s_2)$ forming a claw in $D(CC(G))$.
In this case $CC(G)$ (and $G$) has a bipartite claw (see Figure \ref{fig:theo-charact}(b) and (d));
 \item[5.] with two edges, forming a triangle (same as Item 2.): in this case $G$ has a 2-net;
 \item[6.] with three edges: in this case $CC(G)$ is a diamond (a forbidden subgraph for block graphs) is not a block graph and. So, $G$ is not strictly chordal;
 \item[7.] with two edges, forming a $C_4$: in this case $G$ is not chordal.
\end{enumerate}

Items 1 and 3 satisfy the hypothesis.
With items 2, 4, 5 and 6, by Theorems \ref{theo:charact-SIG1} and Theorem \ref{theo:BW}, we do not have a strictly interval graph.
With Item 7, $G$ is not chordal.
We know already that $D(CC(G))$ is connected; Item (2) shows that it is acyclic.
After the other cases we conclude that $D(CC(G))$ is a path.

\bigskip

$\Longleftarrow$ Let $D(CC(G))$ be a path.
Since $G$ is strictly chordal, we need to prove that $G$ is an interval graph, 
or, equivalently,  $G$ has a clique-tree that is a path (Theorem \ref{theo:charact-IG2})

By Lemma \ref{lem:vertices-cliques}, each critical clique of $G$ consists of all simplicial vertices 
of a maximal clique   or all vertices of a minimal vertex separator of $G$.
So, each maximal clique  of $CC(G)$ corresponds to a maximal clique of $G$.
By hypothesis $D(CC(G))$ is a path, so the maximal cliques established by the minimal vertex separators form a path.
The incident vertices to each edge of $D(CC(G))$ belong to a maximal clique of $G$.
So, these maximal cliques also form a path in the clique tree of $G$.

Let us consider the addition of the simplicial vertices of $CC(G)$ to $D(CC(G))$. 
This addition produces two different cases:

\begin{itemize}
\item[i] the construction of boundary cliques of $CC(G)$, 
containing the minimal vertex separators that are vertices of $D(CC(G))$.
It does not matter how many boundary cliques contain the same $mvs$;
 they will appear sequentially   in the clique-tree of $G$ and the resulting graph  is an interval graph.
 
\item[ii] the addition of a simplicial vertex adjacent to an edge of $D(CC(G))$;
it results in a maximal clique
with three vertices, one of them being a simplicial vertex.
In the clique-tree, this vertex corresponds to an increase of the cardinality of one maximal clique
and it does not affect the clique-tree.
A new addition over this same edge of $D(CC(G))$ cannot be performed: 
the new vertex cannot be a true twin of the first vertex added 
because the maximal cliques of $CC(G)$ have at most one simplicial vertex.
It also cannot be a false twin because the result would be a new maximal clique 
producing a dart in $G$ - impossible since $G$ is a strictly chordal graph.
\end{itemize}

Therefore, there are only two possible additions of simplicial vertices to $D(CC(G))$ (by hypothesis, a path) 
and both maintain the linear order of the maximal cliques of $G$. 
Then, $G$ is an interval graph and the result follows.
\end{proof}

%%%%%%%%%%%%
\begin{figure}[h]
\begin{center}
\begin{tikzpicture} 
[scale=.43,auto=left]
 \tikzstyle{every node}=[circle, draw, 
                        inner sep=1.8pt, minimum width=4pt]
 \tikzstyle{vert}=[circle, draw, fill=black,
                        inner sep=1.7pt, minimum width=4pt]

 % \node  (a) at (2,12) {};
  \node [vert] (b) at (4,12) {};
  \node [vert] (c) at (6,12)  {};
 % \node (d) at (8,12)  {};
  \node [vert] (e) at (5,14) {};
 % \node (f) at (5,16) {};

\node at (3.5,11.2) [draw=none,fill=none] {$s_1$};
\node at (6.5,11.2) [draw=none,fill=none] {$s_2$};
\node at (5.1,15) [draw=none,fill=none] {$s_3$};

\foreach \from/\to in {b/c,b/e,c/e}
    \draw (\from) -- (\to);

 %\node (g) at (13,12) {};
  \node [vert] (h) at (15,12) {};
  \node [vert] (i) at (17,12)  {};
  \node [vert] (j) at (19,12)  {};
 % \node  (k) at (21,12) {};
\node [vert] (m) at (17,14) {};
%\node (n) at (17,16) {};

\node at (15,11.2) [draw=none,fill=none] {$s_1$};
\node at (17,11.2) [draw=none,fill=none] {$s_2$};
\node at (19,11.2) [draw=none,fill=none] {$s_3$};
\node at (16,14.5) [draw=none,fill=none] {$s'$};

\foreach \from/\to in {h/i,i/j,i/m}
    \draw (\from) -- (\to);

\node at (5,10) [draw=none,fill=none] {($a$)};
\node at (17,10) [draw=none,fill=none] {($b$)};
\node at (12,8) [draw=none,fill=none] {$D(CC(G))$};

%%%%%%%%%%%%%%

  \node  (a) at (2,2) {};
  \node [vert] (b) at (4,2) {};
  \node [vert] (c) at (6,2)  {};
  \node (d) at (8,2)  {};
  \node [vert] (e) at (5,4) {};
  \node (f) at (5,6) {};

\foreach \from/\to in {a/b,b/c,c/d,b/e,c/e,f/e}
    \draw (\from) -- (\to);

 \node (g) at (13,2) {};
  \node [vert] (h) at (15,2) {};
  \node [vert] (i) at (17,2)  {};
  \node [vert] (j) at (19,2)  {};
  \node (k) at (21,2) {};
\node [vert] (m) at (17,4) {};
\node (n) at (17,6) {};

\foreach \from/\to in {g/h,h/i,i/j,j/k,i/m,m/n}
    \draw (\from) -- (\to);

\node at (5,0) [draw=none,fill=none] {($c$)};
\node at (17,0) [draw=none,fill=none] {($d$)};

\node at (12,-1) [draw=none,fill=none] {$CC(G)$};

\end{tikzpicture}
\caption{Theorem 4.1}
\label{fig:theo-charact}
\end{center}
\end{figure}
%%%%%%%%%%%%%%%

\newpage

Theorem \ref{theo:charact-SIG2}  provides the basic concepts for a recognition algorithm.
The input for the algorithm is   the chordal graph $G=(V,E)$, the set $\mathbb Q$  and the set $\mathbb S$.
The steps are:

\bigskip

\begin{enumerate}
\item[\it{Step 1}]: Recognition of $G$ as strictly chordal;
\item[\it{Step 2}]: Determination of $D(CC(G))$: 
    \begin{enumerate}
        \item determination of the vertices of $CC(G)$;
        \item removal of simplicial vertices of $CC(G)$ to obtain $D(CC(G))$;
    \end{enumerate}  
%\vspace{-0.5cm}
\item[\it{Step 3}]: Test if $D(CC(G))$ is a path.
\end{enumerate}

\bigskip

This algorithm is composed by steps that depend on traversal of the sets of minimal vertex separators and maximal cliques. 
However an implementation can be even simpler -- the construction of $CC(G)$ can be omitted.

\underline{Step 1}: Recognition of $G$ as strictly chordal, according to Theorem \ref{theo:caract-BD} (2).
This step has linear time complexity.

\underline{Step 2}: determination of  $D(CC(G))$; 

The vertex set of graph $CC(G)$ consists of critical cliques of $G$.
By Lemma \ref{lem:vertices-cliques}  
we know that all the vertices of a minimal vertex separator establish a sole critical clique;
these critical cliques can be represented by one vertex $s_i \in S_i$ in $CC(G)$, being $S_i$ a $mvs $ of $G$.
The simplicial vertices of each maximal clique of $G$ establish also a critical clique;
this critical clique is simplicial in $CC(G)$.
Observe that these critical cliques can
be ignored, since they do not appear in the derived graph, that is our goal.
So, $D(CC(G))$ is the same as $G[\{s_i, 1 \leq i \leq |{\mathbb S}|\}]$. 
This computation has also linear time complexity since it depends only on
one traverse of the maximal cliques and one traverse of the adjacency lists of $V$.

\underline{Step 3}: Test if $D(CC(G))$ is a path.
This step depends only of testing the number of vertices, edges and the adjacency of the vertices.
This step has linear time complexity.

So, the recognition of a chordal graph as a strictly interval graph has linear time complexity.

%%%%%%%%%%%%%%%%%%%%%%%%%%%%%%%%
%-------------- spectral basic concepts ------------%
%%%%%%%%%%%%%%%%%%%%%%%%%%%%%%%%
\section{Basic Spectral Concepts}

The {\it Laplacian matrix of a graph} $G$ is defined as $L(G) = {\mathcal D}(G) - A(G)$, where ${\mathcal D}(G)$ 
denotes the diagonal degree matrix and $A(G)$ the adjacency matrix of $G$. 
As $L(G)$ is symmetric, there are $n$ real eigenvalues and as $L(G)$ is positive semidefinite, all these eigenvalues are non-negative. 
Moreover, the all-ones vector $\mathbf{1}$ is a null vector for $L(G)$ and $0$ is an eigenvalue.
We denote the eigenvalues of $L(G)$, the {\it Laplacian eigenvalues} of $G$, by $\mu_1(G) \geq \mu_2(G) \geq ... \geq \mu_n(G) =0$.

All different Laplacian eigenvalues of $G$ together with their multiplicities form the {\it Laplacian spectrum} of $G$, denoted by $SpecL(G) = [ \overline{\mu}_1^{m_1}; ...; \overline{\mu}_d^{m_d}]$, where  $1 \leq d \leq n$ is the number of distinct eigenvalues of $L(G)$ and $m_i$ the multiplicity of $\overline{\mu}_i$, $1 \leq i \leq d$.  
A graph is called {\it Laplacian integral} if $SpecL(G)$ consists of integers. 

The next results about linear algebra will be used later. 
Given a $n \times p$ matrix $A$ and a $m \times q$ matrix $B$, the 
{\it Kronecker product} (or tensor product) $A \otimes  B$ is the $nm \times pq$ block matrix:
$$ A \otimes  B =
    \begin{pmatrix}
   a_{1,1}B & a_{1,2}B & \dots & a_{1,p}B \\
     \vdots              &   \vdots              &              & \vdots \\
    a_{n,1}B & a_{n,2}B & \dots & a_{n,p}B \\
    \end{pmatrix}
    $$

Let $M$ be a real  $n \times n$ matrix in the following block form 
$$ M =
    \begin{pmatrix}
   M_{1,1} & \dots    & M_{1,t} \\
    \vdots   &  \ddots     & \vdots \\
    M_{t,1} & \dots       & M_{t,t} \\
    \end{pmatrix}
    $$
where $M_{i,j}$ are $n_i \times n_j$ matrices, $1 \leq i,j \leq t$ and $n=\sum_{r=1}^{t} n_r$. If $b_{i,j}$ is the average row sum of $M_{i,j}$, then $B = B(M) = (b_{i,j})$ is the {\it quotient matrix} of $M$. 
Moreover, if for each pair $i,j$, $M_{i,j}$ has constant row sum, then $B$ is the {\it equitable quotient matrix} of $M$. 
The next result follows from 
\cite[Theorem~3.9.5]{CRS10}.

\begin{theorem}\label{theo:equitable}{\rm\cite{CRS10}}
Given a matrix $M$ and $B$ an equitable quotient matrix of $M$, the spectrum of $B$ is contained in the spectrum of $M$ ($Spec(B) \subseteq Spec(M)$).
\end{theorem}

%%%%%%%%%%%%%%%%%%%
%%%%%%%%   SI-core
%%%%%%%%%%%%%%%%%%%%

\section{A new subclass: $SI$-core graphs}

We  define a new subclass of strictly interval graphs, the {\it strictly interval-core graphs},
or simply {\it $SI$-core graphs}.

A graph $G$ is an $SI$-core graph when it is a strictly interval graph with  exactly two minimal vertex separators;
 these separators have the same cardinality $s$, $s \geq 2$, their union defines a maximal clique of size $2s$
 and each separator has the same number $p \geq 2$ of simplicial vertices adjacent to it.
Given fixed integers $s\geq 2$ and $p \geq 2$, we denote $\mathcal{G}(s,p)$ the set of $SI$-core graphs
with these parameters.

Figure \ref{fig:plain} shows an example of an $SI$-core graph $G \in \mathcal{G}(2,3)$.  

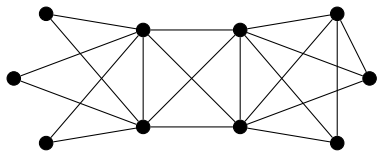
\begin{figure}[h]
\begin{center}
\begin{tikzpicture}
  [scale=.43,auto=left]
\tikzstyle{every node}=[circle, draw, fill=black,
                         inner sep=1.8pt, minimum width=4pt]
%grafo G
\node (a) at (6,-3){};
\node (b) at (6,0){};
\node (f) at (9,-3){};
\node (e) at (9,0){};
\node (g) at (3,0.5){};
\node (h) at (2,-1.5){};
\node (i) at (3,-3.5){};
\node (m) at (12,0.5){};
\node (n) at (13,-1.5){};
\node (o) at (12,-3.5){};
\foreach \from/\to in {a/b, a/e, a/f, b/e, b/f, e/f,  
g/a,g/b,h/a,h/b,i/a,i/b,m/e,m/f,n/e,n/f,o/e,o/f,m/n,g/h,g/i,h/i}  
\draw (\from) -- (\to);

\end{tikzpicture}
\caption{An $SI$-core graph in $\mathcal{G}(2,3)$}
\label{fig:plain}
\end{center}
\end{figure}

\begin{figure}[h]
\begin{center}
\begin{tikzpicture}
  [scale=.43,auto=left]
\tikzstyle{every node}=[circle, draw, fill=black,
                         inner sep=1.8pt, minimum width=4pt]
%grafo G_max
\node (a) at (6,-3){};
\node (b) at (6,0){};
\node (f) at (9,-3){};
\node (e) at (9,0){};
\node (g) at (3,0.5){};
\node (h) at (2,-1.5){};
\node (i) at (3,-3.5){};
\node (m) at (12,0.5){};
\node (n) at (13,-1.5){};
\node (o) at (12,-3.5){};
\foreach \from/\to in {a/b, a/e, a/f, b/e, b/f, e/f,  
g/a,g/b,h/a,h/b,i/a,i/b,m/e,m/f,n/e,n/f,o/e,o/f,h/g,h/i,g/i, m/n,m/o,n/o}  
\draw (\from) -- (\to);
%grafo Gmin
\node (a) at (-9,-3){};
\node (b) at (-9,0){};
\node (f) at (-6,-3){};
\node (e) at (-6,0){};
\node (g) at (-12,0.5){};
\node (h) at (-13,-1.5){};
\node (i) at (-12,-3.5){};
\node (m) at (-3,0.5){};
\node (n) at (-2,-1.5){};
\node (o) at (-3,-3.5){};
\foreach \from/\to in {a/b, a/e, a/f, b/e, b/f, e/f,  
g/a,g/b,h/a,h/b,i/a,i/b,m/e,m/f,n/e,n/f,o/e,o/f}  
\draw (\from) -- (\to);

\end{tikzpicture}
\caption{The $min$ and $max$ $SI$-core graphs in $\mathcal{G}(2,3)$}
\label{fig:minmaxplain}
\end{center}
\end{figure}
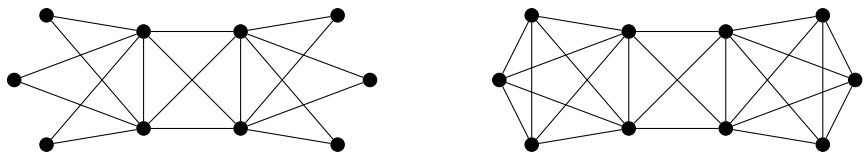

For an element of the new class we will present its Laplacian eigenvalues based on its structure,
showing that some graphs in this class are integral. 
It is interesting to observe that the class of  $SI$-core graphs is $P_5$-free, 
it  is not a subclass of cographs and it is not a subclass of  split graphs.
Actually an $SI$-core graph is a multi-core graph \cite{MP24}.

Two special graphs belonging to $\mathcal{G}(s,p)$ can be considered:
the graphs with minimum and  maximum number of edges, from now on denoted $min$ $SI$-core graph, $G_{min}(s,p)$,
 and $max$ $SI$-core graph, $G_{max}(s,p)$.

Figure \ref{fig:minmaxplain} shows examples of the $min$  and the $max$ $SI$-core graphs in $\mathcal{G}(2,3)$.
Observe that the $min$ $SI$-core graph in Figure \ref{fig:minmaxplain} is the ``carioca" graph,
that belongs to the class defined by Kirkland et al. in \cite{KFDA10}, restudied in \cite{AJM24} as
$(k,t)$-split graphs. 
The $min$ $SI$-core graphs of $\mathcal{G}(s,p)$, for any $s,p$, 
constitute the intersection between the two classes.

\begin{property}\label{edges}
Let $s,p$ be fixed integer values.
The min $SI$-core graph  has $s(2s-1) +2ps$ edges;
the max $SI$-core graph has $s(2s-1) +2ps+ p(p-1)$ edges.
\end{property}

Observe that it is possible to establish, for $s$  and $p$  fixed integers,  the number of non-isomorphic $SI$-core graphs.
In order to do that we need the concept of the partition function $a(n)$. 

A \textit{partition of a positive integer $n$},
also called an integer partition, is a way of writing $n$ as a sum of positive integers.
The number of partitions of $n$ is given by the partition function $a(n)$, entry A000041 in \cite{oeis}.

\begin{theorem}\label{theo:counting}
Let $s,p$ be fixed integer values.
There are $\frac{a(p)(a(p)+1)}{2}$ non-isomorphic $SI$-core graphs.  
\end{theorem}
\begin{proof}
Let $G \in {\mathcal G}(s,p)$.
Each partition of the $p$ simplicial vertices adjacent to a $mvs$ can be combined
with a partition of the other $p$ simplicial vertices.
Hence, with the first partition we have $a(p)$ different combinations.
With the second partition we have $a(p)-1$ different combinations  and so on.
So, the number of combinations without repetition is given by 
the sum of the elements of an arithmetic progression from 1 to $a(p)$. 
\end{proof}

From the results above it is not difficult to see that there are $SI$-core graphs for
any even number of vertices greater or equal 8.
Moreover, we can determine the number of elements of the chosen class.
For instance,  there are 66 non-isomorphic $SI$-core graphs in ${\mathcal G}(4,6)$.

%%%%%%%%%%%%%%%%%%%%%%%%%%%%%%%%%%%%%%%%%%%%%%%%%%%%%%%%%%%%%%%%
%%%%%%   laplacian spectrum 
%%%%%%%%%%%%%%%%%%%%%%%%%%%%%%%%%%%%%%%%%%%%%%%%%%%%%%%%%%%%%%%%%

\section{Laplacian spectrum of $SI$-core graphs}

\begin{theorem}\label{theo:principal}
Let $G_{min}(s,p)$ be the $min$ $SI$-core graph in $\mathcal{G}(s,p)$. 
Then the number $\ell$ of non-integer Laplacian eigenvalues of $G_{min}(s,p)$ with $s \geq 2$ and $p \geq 2$ is $ \ell = 0$ or $ \ell = 2$.
\end{theorem}

\begin{proof}
Let  $G(s,p)$ be the $min$ $SI$-core graph in $\mathcal{G}(s,p)$. 
Then,  $L(G_{min}(s,p))$ can be written as 
$$
    \begin{pmatrix}

    (2s+p) \mathbf{I}_{2s} - \mathbf{J}_{2s} & -\mathbf{X} \\
    -\mathbf{X}^T                            & \mathbf{I}_{2p}
    \end{pmatrix}
   , $$
where $\mathbf{X} = \mathbf{I}_2 \otimes \mathbf{J}_{s,p}$ and $\otimes$ denotes the Kronecker product of matrices.
\newline

It is easy to conclude that $2s +p$ is an eigenvalue of $L(G_{min}(s,p))$ with multiplicity at least $2(s-1)$ and $s$ is an eigenvalue of $L(G_{min}(s,p))$ with multiplicity at least $2(p-1)$ (eigenvalues corresponding to true and false twins). At this point we have $2s+2p-4$ eigenvalues.
The remaining eigenvalues come from the spectrum of the $4 \times 4$ equitable quotient matrix $M$ (Theorem \ref{theo:equitable}):  
$$M =
    \begin{pmatrix}
    (2s+p) \mathbf{I}_2 - s\mathbf{J}_2 & - p \mathbf{I}_2 \\
    - s \mathbf{I}_2            &  s \mathbf{I}_2
    \end{pmatrix}
.$$ 

Observe that the eigenvalues of $M$ are:

\noindent $\bullet$ $0$, corresponding to eigenvector $[\mathbf{1}_2, \mathbf{1}_2]$;\newline
$\bullet$ $s+p$, corresponding to eigenvector $[p\mathbf{1}_2, -s\mathbf{1}_2]$; \newline
$\bullet$ $\lambda$, corresponding to eigenvector $[\alpha z, \beta z]$, 
where $z$ is orthogonal to vector $\mathbf{1}_2$ and $\lambda$ is an eigenvalue of matrix $M_1$ corresponding to eigenvector 
$[\alpha, \beta]$, $M_1 =     \begin{pmatrix}
 (2s+p)  & - p \\
 -s       &   s
    \end{pmatrix}
$.

Matrix $M_1$ has eigenvalues given by  
$\frac{2s+p+s \pm \sqrt{(2s+p-s)^2 + 4sp}}{2}$. 
Since $2s+p+s$ and $2s+p-s$ are strictly positive integer numbers with the same parity, $M_1$ matrix eigenvalues are integer if and only if $(2s+p-s)^2 + 4ps$ is a perfect square. 
Then we can conclude that all eigenvalues of $L(G_{min}(s,p))$ are integer numbers if and only if 
 $(2s+p-s)^2 + 4ps$ is a perfect square ($\ell = 0$) or $L(G_{min}(s,p))$ has exactly two non-integer Laplacian eigenvalues ($\ell = 2$).
\end{proof}

\begin{corollary}\label{cor:valuesGmin}
Let $G_{min}(s,p)$ be the $min$ $SI$-core graph in $\mathcal{G}(s,p)$. 
The Laplacian spectrum of $G_{min}(s,p)$ with $s \geq 2$ and $p \geq 2$ is
\[
\begin{aligned}
\bigg[  \frac{3s+p  +\sqrt{(s+p)^2 + 4sp}}{2};( 2s+p)^{2(s-1)};(s+p); s^{2(p-1)};\\
\frac{3s+p - \sqrt{(s+p)^2 + 4sp}}{2};0\bigg].
\end{aligned}
\]

\end{corollary}

$G_{min}(s,p)$ is Laplacian integral when $(s+p)^2 + 4sp$ is a perfect square. For instance, $G_{min}(2,3)$, $G_{min}(3,10)$, $G_{min}(4,6)$ are Laplacian integral.  Next, by using the concepts given in \cite{CR17}, it will be proved that if for fixed $s,p \geq 2$, $G_{min}(s,p)$ is Laplacian integral, then all $SI$-core graphs in $\mathcal{G}(s,p)$ are Laplacian integral.

Given a connected graph $G=(V(G), E(G))$, let $(F, S)$ be a cluster in $G$, 
where $F$ is the maximal set of false twin vertices and $S$ is the set of the shared neighbors. 
Let $H=(V(H), E(H))$ be any graph such that  $V(H) = F$. 
The graph $G(H) = (V(G(H)), E(G(H)))$ is defined with $V(G(H)) = V(G)$ and $E(G(H)) = E(G) \cup E(H)$.

Let $(F_1,S_1)$ and $(F_2,S_2)$ be a pair of  clusters in a graph $G$.  We say that $(F_1,S_1)$ and $(F_2,S_2)$ are disjoint if $F_1 \cap F_2 = \emptyset$ and $S_1 \cap S_2 = \emptyset$. 
Now, let be a connected graph $G$ with a family of $k \geq 1$ pairwise disjoint clusters $(F_1, S_1), ..., (F_k, S_k)$, with $|S_j|= s_j$, $|F_j| = c_j$, $1 \leq j \leq k$,  and consider a family of graphs $H_1, ..., H_k$ with $|V(H_j)| = c_j$ for all $j$. The graph $G(H_1, ..., H_k)$ is obtained from $G$ by adding the edges in  $H_1, ..., H_k$ as before.
The Laplacian eigenvalues of $H_1, ..., H_k$ are denoted by  $\mu_1(H_j) \geq ...\geq \mu_{c_{j}} (H_j) = 0$, for $ j=1, ..., k$.

Corollary 4.3 in \cite{CR17} shows that the Laplacian eigenvalues of $G(H_1, ..., H_k)$ remain the same, independently of the graphs $H_1, ...., H_k$, with the exception of $|F_1| +.... +|F_k| - k$ of them, and 
$$ \det(\lambda I - L(G(H_1, ..., H_k)) = p_L(\lambda) \Pi_{j=1}^{k}(\Pi_{i=1}^{c_j-1} (\lambda - (s_j + \mu_i(H_j)))  ). $$

Therefore, 
$s_j + \mu_i(H) \in specL(G(H_1, ..., H_k))$, for $1 \leq j \leq k$ and $1 \leq i \leq c_j-1$ 
and the remaining eigenvalues of $G(H_1, ..., H_k)$ are the roots of the polynomial $p_L(\lambda)$ 
of degree  $n-\sum_{j=1}^{k} c_j  + k$, where $ \det(\lambda I - L(G)) = p_L(\lambda) \Pi_{j=1}^{k} (\lambda - s_j)^{c_j-1}$. \newline

\begin{lemma}\label{graphH}
Given $s,p \geq 2$, let $G(s,p) \not = G_{min}(s,p)$ be an $SI$-core graph in $\mathcal{G}(s,p)$, and $H$ the subgraph of $G$ induced by  $E(H) = E(G(s,p)) \backslash E(G_{min}(s,p))$. Then there exists $H = H_1 \cup H_2$ such that $|V(H_1)| = |V(H_2)| = p$ and  $G(s,p) = G_{min}(s,p)(H_1, H_2)$.
\end{lemma}

\begin{proof}
By definition of $SI$-core graphs, $G$ has $2p$ simplicial vertices, such that $p$ of them are adjacent to each vertex of exactly one minimal vertex separator.
As simplicial vertices can be true twins or false twins, we can conclude that the edges in $E(H)$ connect true twins corresponding to simplicial vertices adjacent to one minimal vertex separator of $G(s,p)$. 
Then, there exists $H = H_1 \cup H_2$ where 
$H_j = K_{p^j_1} \cup \dots \cup K_{p^j_{r_j}}$ 
and $p^j_1, .... p^j_{r_j}$ is a partition of $p$ for $j=1, 2$, such that $G(s,p) = G_{min}(H_1, H_2)$. 
\end{proof}

\begin{theorem}\label{theo:principalplainnotmin}
Given $s,p \geq 2$, let $G(s,p)  \not = G_{min}(s,p)$ be an $SI$-core graph in $\mathcal{G}(s,p)$. 
Then the number $\ell$ of non integer Laplacian eigenvalues of $G(s,p)$ is
$ \ell = 0$ or $ \ell = 2$.
\end{theorem}

\begin{proof}
Consider $G_{min}(s,p)$ in $\mathcal{G}(s,p)$, $s \geq 2$ and $p\geq 2$. 
Let $H$ be the graph in the hypothesis of Lemma \ref{graphH}. By using the notation in \cite{CR17}, $G_{min}(s,p)$ has 2 clusters $(F_1, S_1)$ and $(F_2, S_2)$  where $|F_1| = |F_2| = p$ and $|S_1| = |S_2| = s$ are the number of false twins and the number of shared neighbors ($mvs$s), respectively. 
Furthermore, by Lemma \ref{graphH}, $G_{min}(s,p)(H_1, H_2)$ with $H_1 \cup H_2 = H$, 
and each $H_j$ is the union of the complete graphs $K_l$, $l = p^j_1, .... p^j_{r_j}$ containing simplicial vertices adjacent to the minimal vertex separator $S_j$, $ j= 1, 2$.
By using Corollary 4.3 in \cite{CR17} for $G_{min}(s,p)(H_1, H_2)$, its Laplacian eigenvalues coincide with all but $|F_1| + |F_2| - 2 = 2(p-1)$ eigenvalues of $G_{min}(s,p)$. 
 Among the  new  eigenvalues,  when $p^j_i \geq 2$ , $(s_j + p^j_i)$ has multiplicity $(p^j_i -1)$;  
the remaining eigenvalues are equal to $s_j$.
Then, the number of non-integer Laplacian eigenvalues of $G(s,p)$ is the same as that for $G_{min}(s,p)$ and  the result follows.  
\end{proof}

\begin{corollary}\label{cor:valuesGmax}
Given $s,p \geq 2$, let $G_{max}(s,p)$ be the $max$ $SI$-core graph in $\mathcal{G}(s,p)$. The Laplacian spectrum of $G_{max}(s,p)$ is
\[
\begin{aligned}
\bigg[ \frac{3s+p  +\sqrt{(s+p)^2 + 4sp}}{2}; ( 2s+p)^{2(s-1)}; (s+p)^{2p-1}; \\
\frac{3s+p - \sqrt{(s+p)^2 + 4sp}}{2}; 0 \bigg].
\end{aligned}
\]

\end{corollary}
\begin{proof}
Since $G_{max}(s,p)$ verifies the hypothesis of Theorem \ref{theo:principalplainnotmin}, 
and $G_{max}(s,p) = G_{min}(s,p)(H_1, H_2)$ where $H_1 = H_2 = K_p$ as defined below, the only Laplacian eigenvalues of $G_{max}(s,p)$ different from those of $G_{min}(s,p)$ are $(s+p)^{2(p-1)}$. 
By Corollary \ref{cor:valuesGmin}, the result follows.
\end{proof}

\begin{corollary}\label{cor:valuesG}
Given $s,p \geq 2$, let $G(s,p) \not = G_{min}(s,p)$ and $\not = G_{max}$ be an $SI$-core graph in $\mathcal{G}(s,p)$ such that $G(s,p) = G_{min}(s,p)(H_1, H_2)$ where $H = H_1 \cup H_2$.  
The Laplacian spectrum of $G(s,p)$  is
\[
\begin{aligned}
\bigg[ \frac{3s+p  +\sqrt{(s+p)^2 + 4sp}}{2}; ( 2s+p)^{2(s-1)}; s+p; (s+\mu_{1}(L(H));\\ 
...; (s+\mu_{2(p-1)}(L(H))); \frac{3s+p - \sqrt{(s+p)^2 + 4sp}}{2}; 0 \bigg],
\end{aligned}
\]
where 
$\mu_{i}(H)$, $1 \leq i \leq 2(p-1)$ are the $2(p-1)$ eigenvalues with the highest  values.
\end{corollary}
\begin{proof} 
By the proof of Theorem \ref{theo:principalplainnotmin}, the Laplacian eigenvalues of $G(s,p)$ are $2s+2$ Laplacian eigenvalues of $G_{min}(s,p)$ $\bigg(\frac{3s+p  +\sqrt{(s+p)^2 + 4sp}}{2}; ( 2s+p)^{2(s-1)}; s+p;\frac{3s+p - \sqrt{(s+p)^2 + 4sp}}{2}; 0 \bigg)$ 
and the largest $2(p-1)$  Laplacian eigenvalues of $H$, $(s+\mu_{1}(H));...; (s+\mu_{2(p-1)}(H)$).
Observe that the omitted eigenvalues of $H = H_1 \cup  H_2$ in the last expression correspond to zero eigenvalues of $H$ ($\mu_{2p-1}(H) = \mu_{2p}(H)=0$). 
Then, the Corollary follows.
\end{proof}

\bigskip

For instance, for the $SI$-core graph in Figure~\ref{fig:plain}, $G(2,3) = G_{min}(2,3)(H_1, H_2)$ with $H_1 = K_3$ and $H_2 = K_2 \cup K_1$.   $SpecL(H_1) = [3; 3; 0]$, $SpecL(H_2) = [2; 0; 0]$, $SpecL(G_{min}(2,3)) = [8; 7^2; 5; 2^4; 1; 0]$ and 
$SpecL(G(2,3))$ \newline
$ = [8; 7^2; 5^3; 4; 2; 1; 0]$.

\bigskip

The next property follows from Corollaries \ref{cor:valuesGmin},  \ref{cor:valuesGmax} and \ref{cor:valuesG}

\begin{property}\label{prop:sameeigen}
Let $s,p \geq 2$ be fixed integer values.
\begin{itemize}
\item The Laplacian eigenvalues $\mu_1, \mu_2,\dots, \mu_{2s-1}, \mu_{2s},  \mu_{n-1},\mu_n$
are the same for any graph $G\in \mathcal{G}(s,p)$.
\item For $2s+1 \leq i \leq n-2$ the Laplacian eigenvalues $\mu_i$ are integers and they 
belong to the interval $[s,s+p]$ for any graph $G\in \mathcal{G}(s,p)$.
\end{itemize}
\end{property}

%%%%%%%%%%%%%%%%%%%%%%%%%%
%%%%%%%  questions?
%%%%%%%%%%%%%%%%%%%%%%%%%%
\section{Related questions}

In this section we answer some further questions that may arise.
\bigskip

\noindent \underline{Question 1}: Given $n$, is it possible to build an $SI$-core graph with $n$ vertices?

Firstly $n$ must be even;
with any pair $s \geq 2$ and $p \geq 2$ such that $2s+2p=n$
it is possible to build $G \in {\mathcal G}(s,p)$. 
The $2s$ vertices form the maximal clique of $mvs$ $S_1$ and $S_2$;
$p$ vertices are adjacent to $S_1$ and $p$ vertices are adjacent to $S_2$.
These vertices can form mutually exclusive cliques. 

\bigskip

\noindent\underline{Question 2}: Are there Laplacian spectra of $SI$-core graphs such that exactly one $SI$-core graph attains it?

The answer is yes. Given $s, p \geq 2$, the following spectra  

\[
\begin{aligned}
\bigg[  \frac{3s+p  +\sqrt{(s+p)^2 + 4sp}}{2};( 2s+p)^{2(s-1)};(s+p); s^{2(p-1)};\\
\frac{3s+p - \sqrt{(s+p)^2 + 4sp}}{2};0\bigg]
\end{aligned}
\]
and 
\[
\begin{aligned}
\bigg[ \frac{3s+p  +\sqrt{(s+p)^2 + 4sp}}{2}; ( 2s+p)^{2(s-1)}; s+p; (s+p)^{2(p-1)}; \\
\frac{3s+p - \sqrt{(s+p)^2 + 4sp}}{2}; 0 \bigg]
\end{aligned}
\]
are only attained by $G_{min}(s,p)$ and $G_{max}(s,p)$, respectively.

\bigskip

\noindent\underline{Question 3}:

Are there non-isomorphic $SI$-core graphs with the same Laplacian spectrum?

The answer of Question 2 already shows that $G_{min}(s,p)$ and $G_{max}(s,p)$
can be determined by their Laplacian spectrum.
However the same cannot be said for all $SI$-core graphs.
For instance, graphs $G_1$ and $G_2$ in Figure \ref{fig:exampleQ2}
have the same spectrum, 
$$
[10.21699, 9^2, 7, 4^2, 2^6, 0.78301, 0]$$
however they are not isomorphic.

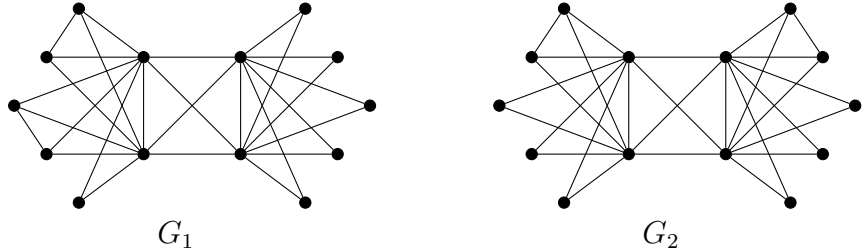
\begin{figure}[h]
\begin{center}
\begin{tikzpicture}
  [scale=.43,auto=left]
\tikzstyle{every node}=[circle, draw, fill=black,
                         inner sep=1.5pt, minimum width=4pt]
%grafo G_max
\node (a) at (6,-3){};
\node (b) at (6,0){};
\node (f) at (9,-3){};
\node (e) at (9,0){};
\node (g) at (3,0){};
\node (g') at (4,1.5){};
\node (h) at (2,-1.5){};
\node (i) at (3,-3){};
\node (i') at (4,-4.5){};
\node (m) at (12,0){};
\node (m') at (11,1.5){};
\node (n) at (13,-1.5){};
\node (o) at (12,-3){};
\node (o') at (11,-4.5){};
\foreach \from/\to in {a/b, a/e, a/f, b/e, b/f, e/f,  
g/a,g/b,h/a,h/b,i/a,i/b,m/e,m/f,n/e,n/f,o/e,o/f,
g'/b,g'/a,i'/a,i'/b,m'/e,m'/f,o'/e,o'/f,
g'/g,m'/m}  
\draw (\from) -- (\to);
  \node  at (7,-5.5) [draw=none,fill=none] {$G_2$}; 
%grafo Gmin
\node (a) at (-9,-3){};
\node (b) at (-9,0){};
\node (f) at (-6,-3){};
\node (e) at (-6,0){};
\node (g) at (-12,0){};
\node (g') at (-11,1.5){};
\node (h) at (-13,-1.5){};
\node (i) at (-12,-3){};
\node (i') at (-11,-4.5){};
\node (m) at (-3,0){};
\node (m') at (-4,1.5){};
\node (n) at (-2,-1.5){};
\node (o) at (-3,-3){};
\node (o') at (-4,-4.5){};
\foreach \from/\to in {a/b, a/e, a/f, b/e, b/f, e/f,  
g/a,g/b,h/a,h/b,i/a,i/b,m/e,m/f,n/e,n/f,o/e,o/f,
g'/b,g'/a,i'/a,i'/b,m'/e,m'/f,o'/e,o'/f,
g'/g,h/i}  
\draw (\from) -- (\to);
  \node  at (-8,-5.5) [draw=none,fill=none] {$G_1$}; 

\end{tikzpicture}
\caption{Two $SI$-core graphs in $\mathcal{G}(2,5)$}
\label{fig:exampleQ2}
\end{center}
\end{figure}

\bigskip

\noindent \underline{Question 4}: Given  $a \in N$,  fixed $s$ and $p$,  is there an $SI$-core graph $G(s,p)$  whose Laplacian spectrum contains $a$?

Since $s$ and $p$ are already known, several eigenvalues of the graph have their value determined in advance.
However, there are $2(p-1)$ values which can be modified.
For these values, we know that some restrictions must be applied.

Firstly $s \leq a \leq s+p$. 

 If $a = s$, the solution is immediate.
Actually, the value $a=s$ can appear $2(p-1)$ times in the spectrum of $G$. Observe $G_{min}(s,p)$.
Its spectrum is already known:

\[
\begin{aligned}
\bigg[  \frac{3s+p  +\sqrt{(s+p)^2 + 4sp}}{2};( 2s+p)^{2(s-1)};(s+p); s^{2(p-1)};\\
\frac{3s+p - \sqrt{(s+p)^2 + 4sp}}{2};0\bigg].
\end{aligned}
\]

Let us consider $a \neq s$, that is, $a > s$.

The term $s^{2(p-1)}$ corresponds to:
$s^{p-1}, s^{p-1}$.
These eigenvalues correspond to simplicial vertices.
To include the value $a$, each eigenvalue must correspond to an element of a  clique of simplicial vertices of size $a-s$.
So, one of these terms can be rewritten as $a^{a-s-1}, s^{p-a+s}$.
For instance, let $a=7$ be the value that  we want to appear as an eigenvalue of $G(3,10)$.
We know that $s < a \leq s+p$.  
The value 7 satisfies the condition $3 < 7 \leq 13$,
so $7 $ must appear exactly  $3 = a - s -1$ times as an eigenvalue.

%%%%%%%%%%%%%%%%%%%%%%%%%%%%%%%%%%%%%%%%%%%%%%%%%%%%%%%%%%%%%%%%%
%%%%%%  acknowledgment
%%%%%%%%%%%%%%%%%%%%%%%%%%%%%%%%%%%%%%%%%%%%%%%%%%%%%%%%%%%%%%%%%

\section{Acknowledgment}The second author was partially supported by CNPq, Grant 405552/2023-8 and FAPERJ, Process SEI 260003/001228/2023.

%%%%%%%%%%%%%%%%%%%%%%%%%%%%%%%%%%%%%%%%%%%%%%%%%%%%%%%%%%%%%%%%%
%%%%%% references
%%%%%%%%%%%%%%%%%%%%%%%%%%%%%%%%%%%%%%%%%%%%%%%%%%%%%%%%%%%%%%%%%

\end{document}